\documentclass{amsart}
\usepackage{amsmath, amssymb, amsthm}
\newtheorem{theorem}{Theorem}[section]
\newtheorem{lemma}[theorem]{Lemma}
\newtheorem{proposition}[theorem]{Proposition}

\theoremstyle{definition}

\newtheorem{example}[theorem]{Example}
\theoremstyle{remark}
\newtheorem{remark}[theorem]{Remark}
\numberwithin{equation}{section}
\begin{document}
\title[R\'enyi Entropy over Countably Infinite Alphabets]{Some Properties of R\'{e}nyi Entropy over\\Countably Infinite Alphabets}
\author{Mladen Kova\v{c}evi\'{c}}
\address{Department of Electrical Engineering, Faculty of Technical Sciences, 
         University of Novi Sad, Trg Dositeja Obradovi\'{c}a 6, 21000 Novi Sad, Serbia}
\email{kmladen@uns.ac.rs}
\author{Ivan Stanojevi\'{c}}
\address{Department of Electrical Engineering, Faculty of Technical Sciences, 
         University of Novi Sad, Trg Dositeja Obradovi\'{c}a 6, 21000 Novi Sad, Serbia}
\email{cet\_ivan@uns.ac.rs}
\author{Vojin \v{S}enk}
\address{Department of Electrical Engineering, Faculty of Technical Sciences, 
         University of Novi Sad, Trg Dositeja Obradovi\'{c}a 6, 21000 Novi Sad, Serbia}
\email{vojin\_senk@uns.ac.rs}
\thanks{This work was supported by the Ministry of Science and 
        Technological Development of the Republic of Serbia 
        (grants No. TR32040 and III44003).}
\date{January 30, 2013.}
\begin{abstract}
In this paper we study certain properties of R\'{e}nyi entropy functionals 
$H_\alpha(\mathcal{P})$ on the space of probability distributions over 
$\mathbb{Z}_+$. Primarily, continuity and convergence issues are addressed. 
Some properties shown parallel those known in the finite alphabet case, while 
others illustrate a quite different behaviour of R\'enyi entropy in the infinite 
case. In particular, it is shown that, for any distribution $\mathcal P$ and any 
$r\in[0,\infty]$, there exists a sequence of distributions $\mathcal{P}_n$ 
converging to $\mathcal{P}$ with respect to the total variation distance, such 
that $\lim_{n\to\infty}\lim_{\alpha\to{1+}} H_\alpha(\mathcal{P}_n) = 
\lim_{\alpha\to{1+}}\lim_{n\to\infty} H_\alpha(\mathcal{P}_n) + r$. 
\end{abstract}
\maketitle
\section{Introduction}
R\'{e}nyi entropies are an important family of functionals defined on the 
space of discrete or continuous probability distributions. They were introduced 
by A.\ R\'{e}nyi \cite{renyi,renyi2} on axiomatic grounds as a generalization of 
Shannon entropy, and have been studied extensively ever since. 
\par For a probability distribution $\mathcal{P}=(p_1,\ldots,p_N)$ 
R\'{e}nyi entropy of order $\alpha$, $\alpha \geq 0$, is defined as 
\begin{equation}
\label{eq_def_finite}
 H_{\alpha}(\mathcal{P})=\frac{1}{1-\alpha}\log \sum_{n=1}^{N}p_{n}^{\alpha}, 
\end{equation}
where it is understood that \cite{renyi} 
\begin{equation}
\label{eq_H1}
 \begin{aligned}
  H_1(\mathcal{P}) &\stackrel{\triangle}{=}\lim_{\alpha\to1}H_\alpha(\mathcal{P})  \\
                   &=-\sum_{n=1}^Np_n\log p_n  \\
                   &=H(\mathcal{P}) 
 \end{aligned}
\end{equation}
which is precisely the Shannon entropy of $\mathcal{P}$. (When $\alpha=0$, 
the convention $0^0=0$ is used. The base of the logarithm in 
\eqref{eq_def_finite}, $b>1$, is arbitrary and will not be specified.) 
Hence, R\'enyi entropy can be thought of as a more fundamental concept 
of which Shannon entropy is an important special case. That this is not 
a mere mathematical generalization has been seen afterwards when R\'enyi 
entropies found applications in many scientific disciplines such as 
information and coding theory \cite{coding,coding2,error,csiszar}, 
statistical physics \cite{stat, stab}, multifractal systems \cite{world}, etc. 
Related concepts of R\'{e}nyi divergence \cite{renyi} (see also, e.g., \cite{erven1,erven2})
and conditional R\'{e}nyi entropy \cite{rate,rategauss}, again appropriate 
generalizations of the corresponding Shannon measures, are also studied 
in a variety of contexts. 
\par We intend here to prove some basic properties of R\'enyi entropies 
of discrete random variables with infinite alphabets. We are partly motivated 
by the recent similar developments in Shannon theory. Namely, in a series of 
papers \cite{ho1,ho2,ho3,ho4} Ho et al.\ have presented an extensive study 
of Shannon information measures over countably infinite alphabets and discussed 
the implications of the results regarding some other well-known information-theoretic 
concepts like error probability, typical sequences, etc. Research on the effects 
of unknown or infinite alphabets on the problems of information theory has always 
been quite active. Source coding is another important example 
(see, e.g., \cite{source1,source2} and the references therein), 
as are Markov chains \cite{kendall} and limit theorems in probability. 
\par R\'enyi entropies are also frequently used when the alphabet is countably 
infinite, e.g., in statistical mechanics where systems with an infinite number 
of particles are often considered \cite{obser}, information sources with an 
infinite number of symbols \cite{inf}, etc., but surprisingly, their basic 
properties over such alphabets are rarely explored in the literature. In most 
textbooks \cite{aczel} and papers, properties of R\'enyi entropies for discrete 
probability distributions are stated and proven only in the finite case. 
However, the behaviour of these functions is in some aspects fundamentally 
different over infinite alphabets and this case needs to be treated separately. 
So, for a probability distribution $\mathcal{P}=(p_1,p_2,\ldots)$ and a real 
parameter $\alpha\geq0$, define \cite{rate}: 
\begin{equation}
\label{eq_def}
  H_\alpha(\mathcal{P})= 
  \begin{cases}
    \frac{1}{1-\alpha}\log\sum\limits_{n=1}^\infty p_n^\alpha\, , & \alpha\neq1  \\
      -\sum\limits_{n=1}^\infty p_n\log p_n\, , & \alpha=1
  \end{cases}
\end{equation}
(Without loss of generality we assume that the alphabet is $\mathbb{Z}_+$ -- the set 
of positive integers.)
We will discuss here mainly the (dis)continuity of $H_\alpha(\mathcal{P})$ with 
respect to $\alpha$ and $\mathcal{P}$. The discontinuity of $H_1(\mathcal{P})$ 
(as well as other Shannon information measures) is thoroughly investigated in 
\cite{ho1}; our findings continue this line of research and give some new insights 
into the general behaviour of information measures. 
\section{Region of convergence}
For any probability distribution with a finite number of probability 
masses, R\'{e}nyi entropy of order $\alpha$ exists for any 
$\alpha \geq 0$. However, in the case of distributions with an 
infinite number of masses the problem of divergence appears. 
Obviously, $H_{\alpha}$ is finite for any $\alpha > 1$ 
because $\sum_{n=1}^{\infty}p_n^{\alpha} < \sum_{n=1}^{\infty}p_n=1$. 
Also, it is easy to see that if $H_{\alpha_0}(\mathcal{P})<\infty$ 
then $H_\alpha(\mathcal{P})<\infty$ for all $\alpha\geq \alpha_0$. 
Call 
\begin{equation}
  \alpha_c(\mathcal{P})=\inf\left\{\alpha\geq0 : H_{\alpha}(\mathcal{P})<\infty\right\}
\end{equation}
the \emph{(R\'{e}nyi's) critical exponent} of the probability 
distribution $\mathcal{P}$. Clearly, $\alpha_c(\mathcal{P})\leq 1$ and 
$H_\alpha (\mathcal{P})=\infty$ for all $\alpha<\alpha_c(\mathcal{P})$. 
It is also interesting to see what happens at $\alpha_c$. 
It turns out that $H_{\alpha_c}(\mathcal{P})$ can converge or 
diverge here, depending on the asymptotics (tail) of the distribution. 
In other words, the \emph{(R\'{e}nyi's) region of convergence} 
of the distribution $\mathcal{P}$, defined by 
\begin{equation}
  \mathcal{R}(\mathcal{P})=\left\{\alpha\geq0 : H_\alpha(\mathcal{P})<\infty\right\}
\end{equation}
is of the form $\mathcal{R}(\mathcal{P})=(\alpha_c(\mathcal{P}),\infty)$ or 
$\mathcal{R}(\mathcal{P})=[\alpha_c(\mathcal{P}),\infty)$. 
Next we give examples of distributions with both kinds of convergence 
regions, for any $\alpha_c\in[0,1]$. In the following, notation 
$x_n\sim y_n$ means $\lim_{n\to\infty}x_n/y_n\in(0,\infty)$.
\begin{example}
  Consider a distribution $\mathcal{P}=(p_1,p_2,\ldots)$ with 
  exponentially decreasing tail $p_n\sim 2^{-n}$. Then for any 
  $\alpha>0$ the sum $\sum_{n=1}^{\infty}2^{-\alpha n}$ converges 
  so that $\alpha_c(\mathcal{P})=0$ and $\mathcal{R}(\mathcal{P})=(0,\infty)$. 
\end{example}
\par Note that any distribution with a finite number of probability 
masses also has $\alpha_c(\mathcal{P})=0$, but the convergence 
region is $\mathcal{R}(\mathcal{P})=[0,\infty)$. 
\begin{example}
  Let $\mathcal{P}=(p_1,p_2,\ldots)$ be a distribution with 
  $p_n\sim n^{-\beta}$, $\beta>1$. Then, since the 
  series $\sum_{n=1}^{\infty}n^{-a}$ converges if and only if $a>1$, it 
  follows that $\sum_{n=1}^\infty p_n^\alpha$ converges if and only 
  if $\alpha\beta>1$. So in this case $\alpha_c(\mathcal{P})=\beta^{-1}$ 
  and the region of convergence is $\mathcal{R}(\mathcal{P})=(\beta^{-1},\infty)$. 
\end{example}
\begin{example}
  Consider a distribution $\mathcal{P}=(p_1,p_2,\ldots)$ with 
  $p_n\sim n^{-\beta}\log^{-2\beta}n$, $\beta>1$. Now for any $\alpha<\beta^{-1}$, 
  $\sum_{n=1}^\infty p_n^\alpha$ diverges because 
  $p_n^\alpha\sim n^{-\alpha\beta}\log^{-2\alpha\beta}n$ is decreasing 
  to zero strictly slower than $n^{-1}$. For $\alpha=\beta^{-1}$ we have 
  $p_n^\alpha\sim n^{-1}\log^{-2}n$ and the corresponding sum converges 
  \cite[Theorem 3.29]{analiza}, as is easily seen from the integral criterion 
  for the convergence of series. 
  So in this case $\alpha_c(\mathcal{P})=\beta^{-1}$ and 
  $\mathcal{R}(\mathcal{P})=[\beta^{-1},\infty)$. 
\end{example}
\par The case $\alpha_c(\mathcal{P})=1$ remains.
\begin{example}
\label{ex_inf_entropy}
  Consider a distribution $\mathcal{P}$ with $p_n\sim n^{-1}\log^{-2}n$. 
  Then $-p_n\log p_n \sim n^{-1}\log^{-1}n$ and therefore (again by the 
  integral criterion) $H(\mathcal{P})=\infty$ so that 
  $\mathcal{R}(\mathcal{P})=(1,\infty)$. ($H(\mathcal{P})=\infty$ implies 
  $H_\alpha(\mathcal{P})=\infty$ for $\alpha<1$ because $-p_n\log p_n$ is 
  bounded from above by $p_n^\alpha$ for all $\alpha<1$ and all sufficiently 
  large $n$.) 
\end{example}
\begin{example}
\label{ex_fin_entropy}
  For the last remaining case, consider $\mathcal{P}$ with 
  $p_n\sim n^{-1}\log^{-3}n$. Now $-p_n\log p_n \sim n^{-1}\log^{-2}n$ which 
  implies $H(\mathcal{P})<\infty$, but $H_\alpha(\mathcal{P})=\infty$ for 
  $\alpha<1$ since $p_n^\alpha$ is bounded from below by $n^{-1}$. We conclude 
  that in this case $\alpha_c(\mathcal{P})=1$ and $\mathcal{R}(\mathcal{P})=[1,\infty)$. 
\end{example}
\par These examples illustrate that the critical exponent of a distribution 
is determined entirely by its asymptotic behaviour. Here is a slightly more precise 
statement.
\begin{lemma}
\label{prop_sim}
  Let $\mathcal{P}=(p_1,p_2,\ldots)$ and $\mathcal{Q}=(q_1,q_2,\ldots)$ be 
  probability distributions. If $p_n\sim q_n$, i.e., if $\lim_{n\to\infty} p_n/q_n\in(0,\infty)$ 
  then $\alpha_c(\mathcal{P})=\alpha_c(\mathcal{Q})$ and 
  $\mathcal{R}(\mathcal{P})=\mathcal{R}(\mathcal{Q})$. 
\end{lemma}
\begin{proof}
  The statement follows immediately by considering $\sum_{n=1}^\infty p_n^\alpha$ and 
  $\sum_{n=1}^\infty q_n^\alpha$, and by using the well-known convergence properties 
  of number series \cite{knopp}. 
\end{proof}
\par The following theorem establishes continuity of $H_\alpha(\mathcal{P})$ 
with respect to $\alpha$ and characterizes its behaviour at $\alpha_c$. 
\begin{theorem}
  For any probability distribution $\mathcal{P}$ over $\mathbb{Z}_+$, 
  $H_\alpha(\mathcal{P})$ is a continuous function 
  in $\alpha$ in its region of convergence. Furthermore, if 
  $\alpha_c(\mathcal{P})$ is the critical exponent of $\mathcal{P}$ 
  and $H_{\alpha_c}(\mathcal{P})=\infty$ then 
  $\lim_{\alpha\to\alpha_c+} H_\alpha(\mathcal{P}) = \infty$. 
\end{theorem}
\begin{proof}
The case $\alpha_c=1$, and continuity at the point $\alpha=1$ for arbitrary 
$\alpha_c$ will be analyzed in Section IV. Assume $\alpha_c<1$, and note 
that in $(\alpha_c,1)\cup(1,\infty)$ it is enough to consider the function 
$\sum_{n=1}^\infty p_n^\alpha$. Since 
all summands are continuous functions in $\alpha$, their sum will 
also be continuous if it converges uniformly \cite[Theorem 7.11]{analiza}, 
so let us check that it does. Assume first that 
$\sum_{n=1}^\infty p_n^{\alpha_c}<\infty$. 
For all $\alpha\geq\alpha_c$, $p_n^\alpha \leq p_n^{\alpha_c}$. 
By Weierstrass' criterion \cite[Theorem 7.10]{analiza} for the uniform 
convergence of functional series, 
these are precisely the sufficient conditions for the 
uniform convergence of $\sum_{n=1}^\infty p_n^\alpha$ on $[\alpha_c,\infty)$ 
and therefore this is a continuous function. 
If $\sum_{n=1}^\infty p_n^{\alpha_c}=\infty$ then one can apply the 
same reasoning with any $\alpha_0>\alpha_c$ instead of $\alpha_c$, 
to establish continuity in $\mathcal{R}(\mathcal{P})$. 
In this case it is left to prove that $H_\alpha(\mathcal{P})$ has 
a vertical asymptote at the critical exponent. This is straightforward, since 
$\lim_{\alpha\to\alpha_c+}\sum_{n=1}^\infty p_n^\alpha 
   \geq \lim_{\alpha\to\alpha_c+}\sum_{n=1}^N p_n^\alpha 
   =    \sum_{n=1}^N p_n^{\alpha_c}$, and the last term can be made arbitrarily 
large by letting $N\to\infty$. 
\end{proof}
By the same arguments as in the finite case \cite{renyi2}, it can be shown 
that $H_\alpha(\mathcal{P})$ is monotonically decreasing in $\alpha$ in its 
region of convergence, unless $\mathcal{P}$ is a uniform distribution in which 
case $H_\alpha(\mathcal{P})$ is constant with respect to $\alpha$. 
\par Let us introduce one more concept related to the R\'{e}nyi's 
convergence region of a distribution. Let $\Gamma$ denote the set 
of all probability distributions over $\mathbb{Z}_+$, 
i.e., $\Gamma=\big\{(p_1,p_2,\ldots): p_n\geq0, \sum_{n=1}^{\infty}p_n=1\big\}$ 
and let $\Gamma(\alpha_c)$ be the set of all distributions with 
critical exponent $\alpha_c$. 
\begin{remark}
 In the following, we will discuss some notions for which a topology on the 
space of all distributions $\Gamma$ is needed. The topology understood in 
this paper is the one induced by the total variation (or variational) distance 
\begin{equation}
 d_{\textsc{tv}}(\mathcal{P},\mathcal{Q}) ={\lVert\mathcal{P}-\mathcal{Q}\rVert}_1  
                              =\sum_{n=1}^\infty \left|p_n-q_n\right|
\end{equation}
where ${\lVert\cdotp\rVert}_1$ is the familiar $\ell^1$ norm. 
\end{remark}
\begin{proposition}
\label{prop_dense}
 $\Gamma(\alpha_c)$ is dense in $\Gamma$, for any $\alpha_c\in [0,1]$. 
In other words, $\Gamma$ is the closure of $\Gamma(\alpha_c)$, 
$\Gamma=\overline{\Gamma(\alpha_c)}$.
\end{proposition}
\begin{proof}
Let $\mathcal{P}=(p_1,p_2,\ldots)$ be an arbitrary distribution in 
$\Gamma$. We need to show that in any neighborhood of $\mathcal P$ 
there exist distributions with critical exponent $\alpha_c$. Assume 
first that $\mathcal{P}$ has an infinite support (i.e., infinitely 
many nonzero masses), and let $n_0$ be such that 
\begin{equation}
 \sum_{n=n_0}^\infty p_n \leq \frac{\epsilon}{2} , 
\end{equation}
where $\epsilon>0$ is an arbitrary small number. 
Let $\mathcal{Q}\in\Gamma(\alpha_c)$ be a distribution with an infinite 
support and critical exponent $\alpha_c$. Take 
$(q_{n_0},q_{n_0+1},\ldots)$ and multiply it by a suitable constant 
to get $(q_{n_0}',q_{n_0+1}',\ldots)$ such that  
\begin{equation}
\label{eq_tail}
 \sum_{n=n_0}^\infty q_n'=\sum_{n=n_0}^\infty p_n.
\end{equation}
Now let $\mathcal{S}=(s_1,s_2,\ldots)$ be a distribution defined by 
\begin{equation}
 \mathcal{S}=(p_1,\ldots,p_{n_0-1},q_{n_0}',q_{n_0+1}',\ldots).
\end{equation}
Clearly $\mathcal{S}\in\Gamma(\alpha_c)$, because $s_n\sim q_n$. 
Furthermore, 
\begin{equation}
 {\lVert\mathcal{P}-\mathcal{S}\rVert}_1 = \sum_{n=1_{\ }}^\infty \left|p_n-s_n\right|  
                                         \leq \sum_{n=n_0}^\infty \left(|p_n|+|s_n|\right)  
                                         = \sum_{n=n_0}^\infty \left(p_n + q_n'\right)  
                                          \leq \epsilon.
\end{equation}
Therefore, in the $\epsilon$-neighborhood of $\mathcal{P}$ we have found 
a member of $\Gamma(\alpha_c)$. Essentially, this completes the proof 
of the claim, but when $\mathcal{P}$ has finite support the proof has 
to be slightly modified (in that case $\mathcal{P}$ has no tail and 
\eqref{eq_tail} fails). So let $\mathcal{P}=(p_1,\ldots,p_N)$ be 
a distribution with finitely many probability masses and $\epsilon>0$ 
an arbitrary small number. Let $\mathcal{Q}\in\Gamma(\alpha_c)$ be a 
distribution with an infinite support and a critical 
exponent $\alpha_c$. Take $(q_{n_0},q_{n_0+1},\ldots)$, $n_0>N$, such that 
\begin{equation}
 \sum_{n=n_0}^\infty q_n \leq \frac{\epsilon}{2}.
\end{equation}
Now create another distribution $\mathcal{S}=(s_1,s_2,\ldots)$ as 
\begin{equation}
 s_n = 
 \begin{cases}
  p_n-\delta_n, & 1 \leq n \leq N  \\
  0,            & N < n < n_0  \\
  q_n,          & n \geq n_0
 \end{cases}
\end{equation}
where $\delta_n$ are such that $p_n-\delta_n\geq0$ and 
$\sum_{n=1}^N \delta_n = \sum_{n=n_0}^\infty q_n$. 
Again, $\mathcal{S}\in\Gamma(\alpha_c)$ because $s_n\sim q_n$. Furthermore, 
\begin{equation}
 {\lVert\mathcal{P}-\mathcal{S}\rVert}_1 = \sum_{n=1}^\infty \left|p_n-s_n\right|  
   = \sum_{n=1}^N \delta_n + \sum_{n=n_0}^\infty q_n  
   \leq \epsilon.
\end{equation}
Therefore, in the $\epsilon$-neighborhood of $\mathcal{P}$ we have found 
a member of $\Gamma(\alpha_c)$. The proof is now complete.
\end{proof}
\begin{proposition}
 $\Gamma(\alpha_c)$ is convex in $\Gamma$, for any $\alpha_c\in [0,1]$. 
\end{proposition}
\begin{proof}
 We need to show that $\mathcal{P},\mathcal{Q}\in\Gamma(\alpha_c)$ implies 
that $\mathcal{T}=\lambda\mathcal{P}+(1-\lambda)\mathcal{Q}\in\Gamma(\alpha_c)$ 
for any $\lambda\in(0,1)$. This is not hard: asymptotic behaviour of 
$\mathcal{T}=\lambda\mathcal{P}+(1-\lambda)\mathcal{Q}$ is determined by 
$\mathcal{P}$ or $\mathcal{Q}$, whichever has heavier tail, so the critical 
exponent is unchanged. More precisely, let $t_n=\lambda p_n + (1-\lambda)q_n$ 
and observe $\sum_{n=1}^\infty t_n^\alpha$. From the fact that $x^\alpha$ is 
a concave and subadditive function for $\alpha<1$ one concludes that 
$\sum_{n=1}^\infty t_n^\alpha$ is also a concave and subadditive function 
in $\mathcal{T}$, i.e.: 
\begin{equation}
 \lambda\sum_{n=1}^\infty p_n^\alpha + (1-\lambda)\sum_{n=1}^\infty q_n^\alpha \;\leq\; \sum_{n=1}^\infty t_n^\alpha 
     \;\leq\; \lambda^\alpha\sum_{n=1}^\infty p_n^\alpha + (1-\lambda)^\alpha\sum_{n=1}^\infty q_n^\alpha. 
\end{equation}
Now it follows that $\sum_{n=1}^\infty t_n^\alpha$ converges if and only if 
$\sum_{n=1}^\infty p_n^\alpha$ and $\sum_{n=1}^\infty q_n^\alpha$ both converge, 
and therefore if $\mathcal{P}$ and $\mathcal{Q}$ have critical exponent $\alpha_c$, 
then so does $\mathcal{T}$. 
\end{proof}
A more general statement can easily be proven by the same methods: For any $\mathcal{P}, 
\mathcal{Q}$, and $\lambda\in(0,1)$, we have that 
$\alpha_c(\mathcal{T})=\max\{\alpha_c(\mathcal{P}),\alpha_c(\mathcal{Q})\}$ and 
$\mathcal{R}_c(\mathcal{T})=\mathcal{R}_c(\mathcal{P})\cap\mathcal{R}_c(\mathcal{Q})$, 
where $\mathcal{T}=\lambda\mathcal{P}+(1-\lambda)\mathcal{Q}$. Still more generally: 
\begin{theorem}
 Let $\mathcal{P}_1,\ldots,\mathcal{P}_K$ be probability distributions over $\mathbb{Z}_+$, 
 and let $\mathcal{T}=\sum_{k=1}^{K}\lambda_k\mathcal{P}_k$ be their mixture, where $\lambda_k>0$ 
 and $\sum_{k=1}^K \lambda_k = 1$. Then $\alpha_c(\mathcal{T})=\max_{1\leq k\leq K}\{\alpha_c(\mathcal{P}_k)\}$ 
 and $\mathcal{R}_c(\mathcal{T})=\bigcap_{k=1}^K \mathcal{R}_c(\mathcal{P}_k)$. 
\end{theorem}
\section{Continuity properties of R\'{e}nyi entropy}
As for the continuity in the argument $\mathcal{P}$, it turns out 
that R\'{e}nyi entropy behaves differently when $\alpha>1$ and when 
$\alpha\leq 1$, unlike its behaviour in the case of finite alphabets.
\begin{theorem}
\label{th_cont}
 The R\'{e}nyi entropy $H_\alpha(\mathcal{P})$ is a continuous function 
in $\mathcal{P}$ for $\alpha>1$ and discontinuous for $\alpha\leq 1$.
\end{theorem}
\begin{proof}
 The discontinuity for $\alpha<1$ can be established as a corollary 
to Proposition \ref{prop_dense}. Take some $\alpha_c>\alpha$. 
In any $\epsilon$-neighborhood of $\mathcal{P}$ there are always 
members of $\Gamma(\alpha_c)$ so we can find a sequence of distributions 
$\mathcal{P}_n\rightarrow\mathcal{P}$ with $\mathcal{P}_n\in\Gamma(\alpha_c)$, 
$\forall n$. In this case, 
$H_\alpha(\mathcal{P}_n)=\infty$ for all $n$ which clearly means that 
$H_\alpha$ is discontinuous. The discontinuity for $\alpha=1$, i.e., 
the discontinuity of Shannon entropy \cite{ho1,wehrl} can be 
proven in a similar way: One can construct a sequence of distributions 
whose entropies diverge by changing the asymptotics of the original 
distribution and staying within a small distance from it. When $\alpha>1$, 
however, R\'{e}nyi entropy is a continuous function. Observe that 
\begin{equation}
 H_\alpha(\mathcal{P})
     =\frac{\alpha}{1-\alpha}\log{\bigg(\sum_{n=1}^\infty p_n^\alpha\bigg)^{\frac{1}{\alpha}}}
     =\frac{\alpha}{1-\alpha}\log{{\lVert\mathcal{P}\rVert}_\alpha} , 
\end{equation}
where ${\lVert\cdotp\rVert}_\alpha$ is the familiar $\ell^\alpha$ norm, 
and hence it is enough to prove the continuity of ${\lVert\mathcal{P}\rVert}_\alpha$. 
It is well known that norm is a continuous function with respect to the metric 
it induces \cite{norma}, i.e., for any sequence of distributions $\mathcal{P}_n$ 
with ${\lVert\mathcal{P}-\mathcal{P}_n\rVert}_\alpha\to0$ we must have 
${\lVert\mathcal{P}_n\rVert}_\alpha\to{\lVert\mathcal{P}\rVert}_\alpha$, 
which follows from the fact that 
\begin{equation}
 {\lVert\mathcal{P}_n-\mathcal{P}\rVert}_\alpha \geq 
        \left|{{\lVert\mathcal{P}_n\rVert}_\alpha - {\lVert\mathcal{P}\rVert}_\alpha}\right|.
\end{equation}
Now continuity with respect to the total variation distance, which we are interested in, 
is easily established by observing that 
\begin{equation}
\label{eq_dist_ineq}
 {\lVert\mathcal{P}-\mathcal{P}_n\rVert}_1 \geq {\lVert\mathcal{P}-\mathcal{P}_n\rVert}_\alpha.
\end{equation}
\end{proof}
\par The following theorem gives more insight into the discontinuity of 
$H_\alpha(\mathcal{P})$ for $\alpha\leq1$. Its special case, for 
$\alpha=1$, is proven in \cite{ho1}.
\begin{theorem}
\label{th_cont2}
 Let $\alpha\in(0,1]$ and let $\mathcal{P}$ be a probability 
distribution over $\mathbb{Z}_+$. Then there exists a sequence 
of distributions $\mathcal{P}_n$ converging to $\mathcal{P}$ with respect to 
the total variation distance, such that 
\begin{equation}
 \lim_{n\to\infty} H_\alpha(\mathcal{P}_n) = H_\alpha(\mathcal{P}) + r
\end{equation}
for arbitrary $r\in[0,\infty]$.
\end{theorem}
\begin{proof}
 The proof for $\alpha=1$ can be found in \cite{ho1}, so assume that $\alpha\in(0,1)$.
The case $r=\infty$ is taken care of by taking $\mathcal{P}_n\in\Gamma(\alpha_c)$ 
for some $\alpha_c>\alpha$, as in the proof of Theorem \ref{th_cont}. In that case 
$H_\alpha(\mathcal{P}_n)=\infty, \forall n$, and so 
$\lim_{n\to\infty}H_\alpha(\mathcal{P}_n)=\infty$. The case $r=0$ 
is trivial, take for example $\mathcal{P}_n=\mathcal{P}$ (but nontrivial 
sequences with $\lim_{n\to\infty} H_\alpha(\mathcal{P}_n) = H_\alpha(\mathcal{P})$ 
can also be constructed). So let $r\in(0,\infty)$. We will construct a sequence 
of distributions $\mathcal{P}_n=(p_{1(n)}, p_{2(n)}, \ldots)$ converging to 
$\mathcal{P}$ and such that 
\begin{equation}
 H_\alpha(\mathcal{P}_n) = H_\alpha(\mathcal{P}) + r
\end{equation}
for all $n$. If $b$ is the base of the logarithm in \eqref{eq_def}, this is equivalent to 
\begin{equation}
\label{eq_b}
 \sum_{i=1}^\infty p_{i(n)}^\alpha = b^{(1-\alpha)r}\sum_{i=1}^\infty p_{i}^\alpha.
\end{equation}
Since $\alpha\in(0,1)$ and $r\in(0,\infty)$, we have $b^{(1-\alpha)r}\in(1,\infty)$. It 
follows that the righthand side of \eqref{eq_b}, call it $h$, satisfies
\begin{equation}
\label{eq_h}
 h > \sum_{i=1}^\infty p_i^\alpha.
\end{equation}
Therefore, we want to construct a sequence $\mathcal{P}_n$ with 
$\sum_{i=1}^\infty p_{i(n)}^\alpha = h$, for arbitrary given $h$ 
satisfying \eqref{eq_h}. The construction is as follows
\begin{equation}
\label{eq_geom}
\begin{aligned}
 \mathcal{P}_n &= (p_{1(n)}, p_{2(n)}, \ldots)  \\
               &= (p_1, \ldots, p_n, B_{(n)}, B_{(n)}q_{(n)}, B_{(n)}q_{(n)}^2, \ldots).
\end{aligned} 
\end{equation}
In other words, we keep the first $n$ probability masses of $\mathcal{P}$ and 
replace the tail of $\mathcal{P}$ with the tail of a geometric distribution. 
According to \eqref{eq_b} and \eqref{eq_geom}, $B_{(n)}$ and $q_{(n)}$ should 
satisfy the following:
\begin{equation}
\label{eq_geom_cond1}
 \sum_{i=0}^\infty B_{(n)}q_{(n)}^i \equiv B_{(n)}\frac{1}{1-q_{(n)}} 
                        = \sum_{i=n+1}^\infty p_i
\end{equation}
and
\begin{equation}
\label{eq_geom_cond2}
 \sum_{i=1}^n p_i^\alpha + \sum_{i=0}^\infty B_{(n)}^\alpha q_{(n)}^{\alpha i} \equiv 
             \sum_{i=1}^n p_i^\alpha + B_{(n)}^\alpha \frac{1}{1-q_{(n)}^\alpha}  
                        = h.
\end{equation}
We need to verify that such $B_{(n)}$ and $q_{(n)}$ exist, i.e., that the above two 
equations have non-negative solutions. Express $B_{(n)}$ from \eqref{eq_geom_cond1} 
\begin{equation}
 B_{(n)}=\left(1-q_{(n)}\right)\sum_{i=n+1}^\infty p_i
\end{equation}
and insert it into \eqref{eq_geom_cond2} to get 
\begin{equation}
\label{eq_solution}
 \frac{\left(1-q_{(n)}\right)^\alpha}{1-q_{(n)}^\alpha} = 
             \frac{h-\sum_{i=1}^n p_i^\alpha}{\left(\sum_{i=n+1}^\infty p_i\right)^\alpha}.
\end{equation}
Now we need to check that the above equation 
has a solution for $q_{(n)}\in(0,1)$ and for all $n\geq n_0$ for some $n_0$. 
To show this observe that the lefthand side is a continuous and monotonically 
increasing function in $q_{(n)}$ starting from $1$ and going to $\infty$ when $q_{(n)}\in(0,1)$. 
This means that \eqref{eq_solution} will have a solution whenever the righthand side 
is greater than $1$. 
This is indeed the case for all $n$ large enough. Namely, the numerator on the 
righthand side of \eqref{eq_solution} tends to $h-\sum_{i=1}^\infty p_i^\alpha$ as $n\to\infty$ 
which is by \eqref{eq_h} strictly positive, and the denominator tends to zero so the entire 
righthand side tends to $\infty$ and is therefore greater than $1$ for $n\geq n_0$ 
for some $n_0$. This means that, for all $n$ (large enough), there exist $B_{(n)}>0$ 
and $q_{(n)}\in(0,1)$ such that \eqref{eq_geom_cond1} and \eqref{eq_geom_cond2} hold. 
Thus we have found a sequence 
$(\mathcal{P}_n)$ with $H_\alpha(\mathcal{P}_n) = H_\alpha(\mathcal{P}) + r$ 
for arbitrary $r\in(0,\infty)$, and, furthermore, from \eqref{eq_geom} and 
\eqref{eq_geom_cond1} it is easy to see that $\mathcal{P}_n\to\mathcal{P}$ 
when $n\to\infty$ with respect to the variational distance. 
We should mention that this proof assumes that $\mathcal{P}$ 
has an infinite support and it needs to be modified when this is 
not true. This is not hard to do but we omit it here (see the proof of 
Proposition \ref{prop_dense} for a similar construction).
\end{proof}
\par Constant $r$ in the previous theorem was taken to be non-negative. 
This is necessary, as the following theorem shows. 
\begin{theorem}
\label{th_lower}
 Let $\mathcal{P}_n, \mathcal{P}$ be probability distributions over $\mathbb{Z}_+$. 
If $\mathcal{P}_n\to\mathcal{P}$ with respect to the total variation distance, then 
$\liminf_{n\to\infty} H_\alpha(\mathcal{P}_n) \geq H_\alpha(\mathcal{P})$. 
\end{theorem}
\begin{proof}
 For $\alpha>1$, $H_\alpha$ is continuous and the claim is obviously true. 
Suppose $\alpha<1$. Let $\mathcal{P}_n=(p_{1(n)}, p_{2(n)} \ldots)$ and 
$\mathcal{P}=(p_1, p_2, \ldots)$, and let $\mathcal{P}_n^{(K)} = (p_{1(n)}, \ldots, p_{K(n)})$, 
$\mathcal{P}^{(K)} = (p_1, \ldots, p_K)$. $\mathcal{P}_n^{(K)}$ and $\mathcal{P}^{(K)}$ 
are obviously not probability distributions but that does not affect the proof. 
For example, $H_\alpha(\mathcal{P}_n^{(K)})$ are well-defined. 
Now, if $\mathcal{P}_n\to\mathcal{P}$ then also $\mathcal{P}_n^{(K)}\to\mathcal{P}^K$ 
when $n\to\infty$. It follows that 
\begin{equation}
\label{eq_partial_distr}
 \lim_{n\to\infty} H_\alpha(\mathcal{P}_n^{(K)}) = H_\alpha(\mathcal{P}^{(K)})
\end{equation}
because R\'{e}nyi entropies are continuous when the alphabet is finite. Now, since 
\begin{equation}
 \sum_{i=1}^\infty p_{i(n)}^\alpha \geq \sum_{i=1}^K p_{i(n)}^\alpha
\end{equation}
and hence (for $\alpha<1$) 
\begin{equation}
 H_\alpha(\mathcal{P}_n) \geq H_\alpha(\mathcal{P}_n^{(K)})
\end{equation}
it follows from \eqref{eq_partial_distr} that 
\begin{equation}
 \liminf_{n\to\infty} H_\alpha(\mathcal{P}_n) \geq H_\alpha(\mathcal{P}^{(K)}).
\end{equation}
This is true for all $K$ and so
\begin{equation}
 \liminf_{n\to\infty} H_\alpha(\mathcal{P}_n) \geq 
                                   \lim_{K\to\infty} H_\alpha(\mathcal{P}^{(K)})  
                        = H_\alpha(\mathcal{P}).
\end{equation}
The case $\alpha=1$ is completely analogous.
\end{proof}
The property stated in Theorem \ref{th_lower} is usually referred to as 
lower-semicontinuity. It is a well known property of Shannon entropy 
\cite{wehrl, ho2} and is now generalized to all R\'{e}nyi entropies. 
Also, the proof is much simpler, in our opinion, than those 
reported before for Shannon entropy.
\par We mention in this context one more property of $H_\alpha$.
\begin{theorem}
 $H_\alpha(\mathcal{P})$ is a $\cap$-convex function in $\mathcal{P}$ for 
$\alpha\leq1$ and is neither $\cap$- nor $\cup$-convex for $\alpha>1$.
\end{theorem}
This is proven in \cite{error} and those arguments easily transfer to the 
infinite case.
\section{The limiting case $\alpha \rightarrow 1$}
Now let us consider the behaviour of R\'{e}nyi entropy at the point $\alpha=1$. For a fixed 
finite alphabet, R\'{e}nyi entropy is defined at this point \eqref{eq_H1} so 
as to preserve continuity (in $\alpha$) \cite{renyi}. There are  
several issues in the case of an infinite alphabet which make continuity 
more difficult to prove than in the finite case. First, 
$H(\mathcal{P})$ might be infinite (see Example \ref{ex_inf_entropy}), and in that 
case it needs to be checked how $H_\alpha(\mathcal{P})$ behaves 
as $\alpha\to 1+$. Next, it is possible that $H(\mathcal{P})<\infty$ 
but $H_\alpha(\mathcal{P})=\infty$ for all $\alpha<1$ 
(see Example \ref{ex_fin_entropy}) 
in which case clearly $\alpha\to 1$ needs to be separated into two 
cases $\alpha\to {1-}$ and $\alpha\to {1+}$. And finally, even without 
these two situations, one needs to be careful when interchanging 
limiting operations because infinite sums are involved. 
\begin{theorem}
\label{th_shannon}
 If $\alpha_c(\mathcal{P})<1$ then 
                 $\lim_{\alpha\to 1}H_\alpha(\mathcal{P})=H(\mathcal{P})$.
If $\alpha_c(\mathcal{P})=1$ then 
                 $\lim_{\alpha\to{1+}}H_\alpha(\mathcal{P})=H(\mathcal{P})$.
\end{theorem}
\begin{proof}
 Assume first that $H(\mathcal{P})<\infty$. Then we have 
\begin{align}
\label{eq_lim_step1}
 \lim_{\alpha\to 1+} H_\alpha(\mathcal{P})
 &=  \lim_{\alpha\to{1+}} \frac{\log\sum_{n=1}^{\infty}p_n^\alpha}{1-\alpha}   \\
\label{eq_lim_step2}
 &=  \lim_{\alpha\to{1+}} \frac{\sum_{n=1}^{\infty}{p_n^\alpha\log{p_n}}}
                               {- \sum_{n=1}^{\infty}p_n^\alpha}               \\
\label{eq_lim_step3}
 &=  -\frac{\lim_{\alpha\to{1+}} \sum_{n=1}^{\infty}{p_n^\alpha\log{p_n}}}
           {\lim_{\alpha\to{1+}} \sum_{n=1}^{\infty}p_n^\alpha}                \\
\label{eq_lim_step4}
 &=  -\frac{\sum_{n=1}^{\infty}{\lim_{\alpha\to{1+}} p_n^\alpha\log{p_n}}}
           {\sum_{n=1}^{\infty}{\lim_{\alpha\to{1+}} p_n^\alpha}}              \\
\label{eq_lim_step5}
 &=  \frac{-\sum_{n=1}^{\infty}{p_n\log{p_n}}}
          {\sum_{n=1}^{\infty} p_n}                                            \\
\label{eq_lim_step6}
 &=   H(\mathcal{P}).
\end{align}
Let us justify the above steps. \eqref{eq_lim_step1} is by definition. 
\eqref{eq_lim_step2} follows from L'H\^{o}pital's rule. 
A sufficient condition for its application \cite[Theorem 5.13]{analiza}, 
is the existence of the limit of the ratio of the derivatives which will 
follow from subsequent equations and our assumption $H(\mathcal{P})<\infty$. 
The equality \eqref{eq_lim_step3} is justified by the fact that the limit of the 
denominator is not zero. \eqref{eq_lim_step4} follows from uniform convergence of the 
series $\sum_{n=1}^{\infty}{p_{n}^\alpha\log{p_n}}$ and 
$\sum_{n=1}^{\infty}{p_{n}^\alpha}$ on $[1,\infty)$. This is 
established easily by Weierstrass' criterion \cite{analiza} using the following 
facts (valid for $\alpha>1$) 
\begin{eqnarray}
- p_{n}^\alpha\log{p_n} < - p_{n}\log{p_n},\quad p_n^\alpha < p_{n},   \\
-\sum_{n=1}^{\infty}{p_{n}\log{p_n}}=H(\mathcal{P}),\quad \sum_{n=1}^{\infty} p_{n}=1.
\end{eqnarray}
Steps \eqref{eq_lim_step5} and \eqref{eq_lim_step6} are obvious. 
If $\alpha_c(\mathcal{P})=1$ then clearly the above limit is the only 
one that makes sense. If $\alpha_c(\mathcal{P})<1$ then one can take 
any $\alpha_0\in(\alpha_c,1)$ and repeat the above arguments about 
uniform convergence on $[\alpha_0,\infty)$ and then the claim is true 
when $\alpha\to 1$ (all the other steps are identical). It remains to 
be shown that $\lim_{\alpha\to{1+}} H_\alpha(\mathcal{P})=\infty$ when 
$H(\mathcal{P})=\infty$. To prove this we define a sequence of distributions 
\begin{equation}
\label{eq_distrib_Qn}
 \mathcal{Q}_n=(q_{1(n)},\ldots,q_{n(n)})=(p_1,\ldots,p_{n-1},\sum_{i=n}^{\infty}p_i).
\end{equation}
We have $\lim_{n\to\infty}\mathcal{Q}_n=\mathcal{P}$ 
in the sense that variational distance between $\mathcal{Q}_n$ an 
$\mathcal{P}$ tends to zero. Also 
\begin{equation}
\label{eq_lim_HQn}
 \lim_{n\to\infty} H(\mathcal{Q}_n)=\infty=H(\mathcal{P}).
\end{equation}
This follows from the fact that Shannon entropy is lower-semicontinuous \cite{wehrl}, 
namely $\liminf_{n\to\infty} H(\mathcal{Q}_n)\geq H(\mathcal{P})$ 
(in general, however, $\mathcal{Q}_n\to\mathcal{P}$ does not imply 
$H(\mathcal{Q}_n)\to H(\mathcal{P})$). Observe that, for $\alpha>1$, 
\begin{equation}
\label{eq_temp}
 \bigg(\sum_{i=n}^{\infty}p_i\bigg)^\alpha \geq \sum_{i=n}^{\infty} p_i^\alpha 
\end{equation}
because the function $x^\alpha$ is superadditive for $\alpha>1$. 
Now \eqref{eq_distrib_Qn} and \eqref{eq_temp} give 
\begin{equation}
 \sum_{i=1}^{n} q_{i(n)}^\alpha \geq \sum_{i=1}^{\infty} p_i^\alpha
\end{equation}
which implies
\begin{equation}
 H_\alpha(\mathcal{Q}_n) \leq H_\alpha(\mathcal{P}).
\end{equation}
This is true for any $\alpha>1$ and all $n\geq 1$. Taking limits on both 
sides we get 
\begin{equation}
 \lim_{\alpha\to 1+} H_\alpha(\mathcal{Q}_n) \leq 
                                     \lim_{\alpha\to 1+} H_\alpha(\mathcal{P})
\end{equation}
which holds for all $n$. Now, since $H(\mathcal{Q}_n)<\infty$, 
by the first part of our proof the lefthand side is equal to 
$H(\mathcal{Q}_n)$. And since $H(\mathcal{Q}_n)$ is unbounded \eqref{eq_lim_HQn}, 
the righthand side must be unbounded too, i.e., 
$\lim_{\alpha\to{1+}} H_\alpha(\mathcal{P})=\infty$. 
This completes the proof of the theorem.
\end{proof}
\par Note that, unlike in the case of finite and fixed alphabet, one cannot claim 
that $\lim_{\alpha\to1}H_\alpha(\mathcal{P})=H(\mathcal{P})$ 
but only $\lim_{\alpha\to1+}H_\alpha(\mathcal{P})=H(\mathcal{P})$. 
\par Here are also two, potentially useful, restatements of the above 
theorem (just omit the logarithms). For any sequence 
$(p_1,p_2,\ldots)$, $p_n\geq0$, $\sum_{n=1}^\infty p_n = 1$, 
\begin{equation}
 \lim_{\epsilon\to0+} \left(\sum_{n=1}^{\infty} p_n^{1+\epsilon}\right)^\frac{1}{\epsilon} = 
             \prod_{n=1}^{\infty} {p_n}^{p_n},
\end{equation}
or
\begin{equation}
 \lim_{\alpha\to{1+}} {\lVert\mathcal{P}\rVert}_\alpha^\frac{\alpha}{\alpha-1} = 
             \prod_{n=1}^{\infty} {p_n}^{p_n} . 
\end{equation}
\par Let us exemplify one consequence of these results. Let 
$\mathcal{P}$ be some distribution over a countably infinite alphabet 
such that $H(\mathcal{P})<\infty$. In \cite{ho1}, it is shown that there 
always exists a sequence of distributions $\mathcal{P}_n$ such that 
$\mathcal{P}_n\to \mathcal{P}$, but $H(\mathcal{P}_n)\nrightarrow H(\mathcal{P})$. 
Actually, it is shown \cite[Theorem 2]{ho1}, that for any $c\geq0$, there is 
such a sequence $\mathcal{P}_n$ so that 
$\lim_{n\to\infty} H(\mathcal{P}_n)=H(\mathcal{P})+c$ (this is a special case 
of Theorem \ref{th_cont2} above). Using this and Theorem \ref{th_shannon}, one concludes that 
$\lim_{n\to\infty}\lim_{\alpha\to{1+}} H_\alpha(\mathcal{P}_n)=
\lim_{n\to\infty}H(\mathcal{P}_n)$ need not equal $H(\mathcal{P})$. 
On the other hand, Theorems \ref{th_cont} and \ref{th_shannon} guarantee that  
$\lim_{\alpha\to{1+}}\lim_{n\to\infty} H_\alpha(\mathcal{P}_n)= 
\lim_{\alpha\to{1+}}H_\alpha(\mathcal{P})=H(\mathcal{P})$ 
for any sequence $\mathcal{P}_n\to\mathcal{P}$. We summarize this in 
the form of a theorem whose proof we have essentially described.
\begin{theorem}
\label{th_main}
 Let $\mathcal{P}=(p_1,p_2,\ldots)$ be a probability distribution. 
Then, for any $r\in[0,\infty]$, there exists a sequence of distributions 
$\mathcal{P}_n$ converging to $\mathcal{P}$ with respect to the total 
variation distance, such that 
\begin{equation}
 \lim_{n\to\infty}\lim_{\alpha\to{1+}} H_\alpha(\mathcal{P}_n) = H(\mathcal{P}) + r,
\end{equation}
but for any such sequence
\begin{equation}
 \lim_{\alpha\to{1+}}\lim_{n\to\infty} H_\alpha(\mathcal{P}_n) = H(\mathcal{P}). 
\end{equation}
\end{theorem}
\par In applied sciences one usually freely interchanges limiting operations, 
such as limits, sums, integrals, derivatives, etc. Such rules, however, do not 
always apply, and the above theorem provides an illustrative example of this, 
involving quantities with physical meaning.
\section{The limiting case $\alpha \rightarrow \infty$}
There is one more interesting limiting case for R\'{e}nyi entropies, 
namely $\alpha\to\infty$. It is known \cite{aczel, inflimit} that 
\begin{equation}
 \lim_{\alpha\to\infty}H_\alpha(\mathcal{Q})=-\log\max_n q_n 
\end{equation}
when $\mathcal{Q}$ has finite support. It is easy to prove that this
remains true for any $(q_1,\ldots,q_n)$, $q_i\geq0$, with 
$\sum_{i=1}^n q_i$ not necessarily equal to $1$. The same is true in 
the infinite case, the proof is just a little more subtle. 
\par Let $\mathcal{P}=(p_1,p_2,\ldots)$ be a probability distribution. First observe that 
\begin{align}
 \lim_{\alpha\to\infty} H_\alpha(\mathcal{P})&= 
          \lim_{\alpha\to\infty} \frac{1}{1-\alpha}\log\sum_{i=1}^{\infty} p_i^\alpha   \\
\label{eq_lim_inf_step2}
 &= \lim_{\alpha\to\infty} \frac{\sum_{i=1}^\infty p_i^\alpha\log p_i}
                                {-\sum_{i=1}^\infty p_i^\alpha}  \\
\label{eq_lim_inf_step3}
 &\geq \lim_{\alpha\to\infty} \frac{-\log\max_{i} p_i \sum_{i=1}^\infty p_i^\alpha}
                                   {\sum_{i=1}^\infty p_i^\alpha}  \\
\label{eq_lim_inf_step4}
 &= -\log\max_{i} p_i
\end{align}
where \eqref{eq_lim_inf_step2} is by L'H\^{o}pital's rule and 
\eqref{eq_lim_inf_step3} by lower bounding $-\log p_i$. 
Now to prove that this is also an upper bound, write 
\begin{equation}
 \sum_{i=1}^\infty p_i^\alpha \geq \sum_{i=1}^n p_i^\alpha
\end{equation}
which is true for all $n$ and all $\alpha>0$. Let 
$\mathcal{P}_n=(p_1,\ldots,p_n)$. ($\mathcal{P}_n$ is not a 
probability distribution but that does not affect the proof.) 
Then for $\alpha>1$ we have
\begin{equation}
 H_\alpha(\mathcal{P}) \leq H_\alpha(\mathcal{P}_n),
\end{equation}
and hence 
\begin{equation}
\label{eq_lim_inf_partial}
 \lim_{\alpha\to\infty} H_\alpha(\mathcal{P}) 
           \leq \lim_{\alpha\to\infty} H_\alpha(\mathcal{P}_n)  
           = -\log\max_{i\in\{1,\ldots,n\}} p_i.
\end{equation}
Since \eqref{eq_lim_inf_partial} holds for all $n$, it follows that 
\begin{equation}
 \lim_{\alpha\to\infty} H_\alpha(\mathcal{P}) \leq -\log\max_i p_i.
\end{equation}
Together with \eqref{eq_lim_inf_step4} this yields 
\begin{equation}
 \lim_{\alpha\to\infty} H_\alpha(\mathcal{P}) = -\log\max_i p_i.
\end{equation}
Therefore, R\'{e}nyi entropy of order $\infty$ is well-defined by 
$H_\infty(\mathcal{P})\stackrel{\triangle}{=}\lim_{\alpha\to\infty} H_\alpha(\mathcal{P})$. 
It can easily be shown that $H_\infty(\mathcal{P})$ is continuous in $\mathcal{P}$, 
thus extending Theorem \ref{th_cont}. 
\section{Conclusions}
We have established some basic properties of R\'{e}nyi entropies $H_\alpha(\mathcal{P})$ 
on the space of probability distributions with countably infinite alphabets. 
In the first place, continuity issues with respect to $\alpha$ and $\mathcal{P}$ 
were addressed. We have shown that these properties are partly similar and partly 
different from the finite alphabet case. The differences are perhaps best summarized 
in Theorem \ref{th_main} which shows that one must be very careful when dealing 
with infinities and limiting operations: Even if $\mathcal{P}_n\rightarrow\mathcal{P}$ 
and $H_\alpha(\mathcal{P}_n)<\infty$ for all $n$ and all $\alpha$, one can have 
$\lim_{n\to\infty}\lim_{\alpha\to{1+}} H_\alpha(\mathcal{P}_n) = 
\lim_{\alpha\to{1+}}\lim_{n\to\infty} H_\alpha(\mathcal{P}_n) + r$, 
for arbitrary $r\in[0,\infty]$. 
\end{document}